\documentclass{article}
\usepackage{amsmath}
\usepackage{amssymb}
\usepackage{amsfonts}
\usepackage{amsthm}
\usepackage{bm}
\usepackage[colorlinks,citecolor=blue]{hyperref}
\usepackage{color}
\usepackage{titlesec}
\usepackage{enumerate}

\titleformat{\section}{\Large\bfseries}{\thesection .}{0.4em}{}
\titleformat{\subsection}{\large\bfseries}{\thesubsection .}{0.4em}{}

\numberwithin{equation}{section}

\theoremstyle{plain}
\newtheorem{theorem}{Theorem}

\newtheorem{lemma}{Lemma}

\allowdisplaybreaks

\begin{document}

 
\title{Inductive Approach to Loop Vertex Expansion}  
\author{Fang-Jie Zhao}  
\date{}      
\maketitle 

\begin{abstract}

{\noindent
An inductive realization of Loop Vertex Expansion is proposed and is applied to  
the construction of the $\phi_1^4$ theory. It appears simpler and more natural than the standard one
at least for some  situations.
}
\end{abstract}
\section{Introduction}

Loop Vertex Expansion (LVE) \cite{R07,R10,RW14} is a new constructive method for bosonic field theories, based on the
auxiliary field representation and BKAR forest formula \cite{BK87,AR95}. Contrary to the traditional techniques such as
cluster expansion and Mayer expansion, it needs no space-time discretization at all and thus greatly simplifies
the constructions.

However, the standard LVE requires some modifications to adapt to different situations \cite{RW10,GR14,RW15,LR16}. In \cite{MR08},
to extract the exponential decay of connected Schwinger functions, one performs some expansion on the resolvent and integration by parts with respect to the auxiliary field,  which destroy the simplicity of LVE to a certain extent. 
Although an extra large/small field expansion \cite{AR97} can be used to replace the integration by parts, as point out in \cite{MR08}, it unfortunately reintroduces some discretization of space-time and appears rather heavy. 
 
In this paper, we provide an inductive realization of LVE, which also relies on the use of the auxiliary field. 
However, instead of explicitly applying the BKAR forest formula, 
we derive a functional integral equation for the theory with the auxiliary field, 
in which the tree structure and the combinatorial core are hidden. 
After a proper decomposition of the physical quantities, 
the functional integral equation turns into a series of normal equations and then can be solved inductively.
One of the advantages of our method is that some finer tree structures can be obtained naturally from further decompositions of the physical quantities. 
Also we give a pointwise estimate for the resolvent instead of a norm one and then totally avoid the resolvent expansion.
These two points seemly lead to a simpler proof of the exponential decay of connected Schwinger functions.
As an example, we construct two point Schwinger function of the one dimensional $\phi^4$ theory in this new way 
and leave the 2n-point functions with $n\!>\!1$ to a future paper. 
We also consider the pressure of the theory in order to check the flexibility of this method.

\section{Description}
\noindent
Let us consider the one-dimensional $\phi^4$ theory ($\phi^4_1$).
For mathematical rigor, we at first restrict the interacting party of the theory to the finite interval $[-L,L]$ and then take the limit as $L\!\to\! \infty$ in some way.
The generating functional of the restricted $\phi^4_1$ theory is
\begin{equation}
Z_{L}[J]=\int\! d\mu_C(\phi)\, e^{-\lambda\!\int_{-L}^{L}\!dx\,\phi_x^4+\int\!dx\, J_x\phi_x}
\end{equation}
with the source $J$ in the  Schwartz space $\mathcal{S}(\mathbb{R})$, where
\begin{equation}
C_{x,y}=((-\Delta\!+\!M^2)^{-1})_{x,y}=\frac{1}{2M} e^{-M|x-y|}
\end{equation}
and $d\mu_C$ is the Gaussian measure on the space of tempered distributions $\mathcal{S}'(\mathbb{R})$ with covariance $C$. 
Introducing an auxiliary source $\sigma$ and a parameter $t\!\in\![0,1]$ into the theory, we have
\begin{equation}
Z_{t,L}[J,\sigma]=\int\! d\mu_C(\phi)\, e^{-t\lambda\!\int_{-L}^L\!dx\,\phi_x^4+\int\!dx\, J_x\phi_x-i\sqrt{2\lambda}\!\int\!dx\, \sigma_x\phi_x^2},
\end{equation}
which reduces to $Z_L[J]$ when $t\!=\!1$ and $\sigma\!=\!0$. 
Also we have
\begin{align}\label{partialZ}
\frac{\partial}{\partial t} Z_{t,L}[J,\sigma]&=\!\int\! d\mu_C(\phi)\, e^{-t\lambda\!\int_{-L}^L\!dx\,\phi_x^4+\int\!dx\, J_x\phi_x-i\sqrt{2\lambda}\!\int\!dx\, \sigma_x\phi_x^2}\notag\\
&\times(-\lambda)\!\int_{-L}^L\!dx\,\phi_x^4=\frac12\int_{-L}^L\!dx\,\frac{\delta^2}{\delta\sigma_x^2}Z_{t,L}[J,\sigma]
\end{align}
with the initial condition
\begin{align}
Z_{0,L}[J,\sigma]&=\int\! d\mu_C(\phi)\, e^{\int\!dx\, J_x\phi_x-i\sqrt{2\lambda}\!\int\!dx\, \sigma_x\phi_x^2}\notag\\
&=e^{\frac12 \!\int\!dx\!\int\!dy\, J_x G[{\sigma}]_{x,y}J_y-\frac12\mathrm{Tr}\ln(I+2i\sqrt{2\lambda}A[{\sigma}])},
\end{align}
where $A[{\sigma}]_{x,y}\!=\!\int\! dz\, C^{1/2}_{x,z}\sigma_z C^{1/2}_{z,y}$ and $G[{\sigma}]\!=\!(-\Delta\!+\!M^2\!+\!2i\sqrt{2\lambda}\sigma)^{-1}$. It is easy to see that $G[{\sigma}]$ remains valid  for $|\mathrm{Arg}\, M^2| \!<\! \frac{\pi}{2}$. 

Denoting 
$W_{t,L}[J,\sigma]\!=\!\ln Z_{t,L}[J,\sigma]$, 
we obtain a Polchinski-type equation 
\begin{align}
\frac{\partial}{\partial t} W_{t,L}[J,\sigma]=\frac12\int_{-L}^L\!dx\,\frac{\delta^2}{\delta\sigma_x^2}W_{t,L}[J,\sigma]+\frac12\int_{-L}^L\!dx\,\Big(\frac{\delta}{\delta\sigma_x}W_{t,L}[J,\sigma]\Big)^2
\end{align}
with the initial condition
\begin{align}\label{InitialCondition}
W_{0,L}[J,\sigma]=\frac12 \!\int_{\mathbb{R}}\!dx\!\int_{\mathbb{R}}\!dy\, J_x G[{\sigma}]_{x,y}J_y-\frac12\mathrm{Tr}\ln(I\!+\!2i\sqrt{2\lambda}A[{\sigma}]).
\end{align}
(Actually, $W_{0,L}[J,\sigma]$ is independent of $L$.)
As pointed out in \cite{BK87,BF93}, this equation can be rewritten as a functional integral equation
\begin{align}\label{IntegralEquation}
&W_{t,L}[J,\sigma]=\int\!d\nu_{t,L}(\bar{\sigma})\,W_{0,L}[J,\sigma\!+\!\bar{\sigma}]\,+\notag\\
&\frac12\!\int_{0}^t\!ds\!\int_{-L}^L\!dx\!\int\!d\nu_{t-s,L}(\bar{\sigma})\,\Big(\frac{\delta}{\delta\bar{\sigma}_x}W_{s,L}[J,\sigma\!+\!\bar{\sigma}]\Big)^2,
\end{align}
where $d\nu_{t,L}$ is the Gaussian measure on $\mathcal{S}'(\mathbb{R})$ with covariance $t\delta(x\!-\!y)\chi_{[-L,L]}(x)$ for $t\!\ge\! 0$.
Let $W^{(0)}_{t,L}[J,\sigma]=\int\!d\nu_{t,L}(\bar{\sigma})\,W_{0,L}[J,\sigma\!+\!\bar{\sigma}]$ and for $n\!\ge\!1$
\begin{align}
W^{(n)}_{t,L}[J,\sigma]=\frac12\sum_{h=0}^{n-1}\int_{0}^t\!ds\!
\int_{-L}^L\!dx\!\int\!d\nu_{t-s,L}(\bar{\sigma})\notag\\\frac{\delta}{\delta\bar{\sigma}_x}W^{(h)}_{s,L}[J,\sigma\!+\!\bar{\sigma}]\frac{\delta}{\delta\bar{\sigma}_x}W^{(n-h-1)}_{s,L}[J,\sigma\!+\!\bar{\sigma}].
\end{align}
It is obvious that $W_{t,L}[J,\sigma]\!=\!\sum_{n\ge 0}W^{(n)}_{t,L}[J,\sigma]$ satisfies \eqref{IntegralEquation} formally. 

We now go to the full $\phi_1^4$ theory. Let $d\nu_{t}$ be the Gaussian measure on $\mathcal{S}'(\mathbb{R})$ with covariance $t\delta(x\!-\!y)$ for $t\!\ge\! 0$, which is the formal limit of $d\nu_{t,L}$ as $L\!\to\! \infty$. Let $\nu_{t}\!*\!F$ be an abbreviation for $\int\!d\nu_{t}(\bar{\sigma})\,F[\cdot\!+\!\bar{\sigma}]$.
Denote
\begin{align}
W^{(n)}_{t,y_1,\dots,y_M;z_1,\dots,z_N}[\sigma]&=\lim_{L\to\infty}\prod_{j=1}^M\frac{\delta}{\delta J_{y_j}}\prod_{k=1}^N\frac{\delta}{\delta \sigma_{z_k}} W^{(n)}_{t,L}[J,\sigma]\Big|_{J=0}
\end{align}
for $M\!+\!N\!>\!0$, if the limit exists. 
Then the Schwinger functions of the full theory can be decomposed into $\sum_{n\ge 0}W^{(n)}_{1,y_1,\dots,y_M}[0]$ at least formally and, by the $\mathbb{Z}_2$ symmetry of the $\phi$-interactions, they vanish for $M$ odd.
Here we consider the case $M\!=\!2$ (i.e. the two pint function) only for simplicity and postpone the cases $M\!=\!4,6,\dots$ to a future paper. 
For $n\!=\!0$, direct calculation shows 
\begin{align}
W^{(0)}_{t;z}[\sigma]=-i\sqrt{2\lambda}\,(\nu_{t}\!*\!G_{z,z})[\sigma],\;
W^{(0)}_{t,y_1,y_2}[\sigma]=(\nu_{t}\!*\!G_{y_1,y_2})[\sigma]
\end{align}
and then, in generally,
\begin{align}
&W^{(0)}_{t;z_0,z_1,\dots,z_N}[\sigma]=\frac12(-2i\sqrt{2\lambda})^{N+1}\sum_{\tau\in S_{N}}\notag\\
&\big(\nu_{t}\!*\!\big(G_{z_0,z_{\tau(1)}}G_{z_{\tau(1)},z_{\tau(2)}}\!\cdots G_{z_{\tau(N)},z_0}\big)\big)[\sigma],\\
&W^{(0)}_{t,y_1,y_2;z_1,\dots,z_N}[\sigma]=(-2i\sqrt{2\lambda})^{N}\sum_{\tau\in S_{N}}\notag\\
&\big(\nu_{t}\!*\!\big(G_{y_1,z_{\tau(1)}}G_{z_{\tau(1)},z_{\tau(2)}}\!\cdots G_{z_{\tau(N)},y_2}\big)\big)[\sigma],
\end{align}
where $S_{N}$ is the set of all permutations of $\{1,\dots,N\}$.
For $n\!\ge\!1$, we have, at least formally,
\begin{align}
&W^{(n)}_{t;z}[\sigma]=\lim_{L\to\infty}\sum_{h=0}^{n-1}\int_{0}^t\!ds\!\int_{-L}^L\!dx\!\int\!\!d\nu_{t-s,L}(\bar{\sigma})\notag\\
&\frac{\delta}{\delta\bar{\sigma}_x}W^{(h)}_{s,L}[0,\sigma\!+\!\bar{\sigma}]\frac{\delta}{\delta\bar{\sigma}_z}\frac{\delta}{\delta\bar{\sigma}_x}W^{(n-h-1)}_{s,L}[0,\sigma\!+\!\bar{\sigma}]\notag\\
&=\sum_{h=0}^{n-1}\int_{0}^t\!ds\!\int_{\mathbb{R}}\!dx\,\big(\nu_{t-s}\!*\!\big(W^{(h)}_{s;x}W^{(n-h-1)}_{s;z,x}\big)\big)[\sigma]
\end{align}
and similarly have
\begin{align}
W^{(n)}_{t,y_1,y_2}[\sigma]=\sum_{h=0}^{n-1}\int_{0}^t\!ds\!
\int_{\mathbb{R}}\!dx\,
\big(\nu_{t-s}\!*\!\big(W^{(h)}_{s;x}W^{(n-h-1)}_{s,y_1,y_2;x}\big)\big)[\sigma].
\end{align}
Also, in generally, we get
\begin{align}
&W^{(n)}_{t;z_0,z_1,\dots,z_N}[\sigma]
=\sum_{h=0}^{n-1}\sum_{K\subset \widehat{N}}\int_{0}^t\!ds\!\int_{\mathbb{R}}\!dx\notag\\
&\Big(\nu_{t-s}\!*\!\Big(W^{(h)}_{s;x,(z_k)_{k\in K}}W^{(n-h-1)}_{s;z_0,x,(z_k)_{k\in \widehat{N}\backslash K}}\Big)\Big)[\sigma]
\end{align}
and
\begin{align}
&W^{(n)}_{t,y_1,y_2;z_1,\dots,z_N}[\sigma]=\sum_{h=0}^{n-1}\sum_{K\subset \widehat{N}}\int_{0}^t\!ds\!
\int_{\mathbb{R}}\!dx\notag\\
&\Big(\nu_{t-s}\!*\!\Big(W^{(h)}_{s;x,(z_k)_{k\in K}}W^{(n-h-1)}_{s,y_1,y_2;x,(z_k)_{k\in \widehat{N}\backslash K}}\Big)\Big)[\sigma],
\end{align}
where $\widehat{N}$ is an abbreviation for $\{1,2,\dots,N\}$.

\begin{theorem}\label{theorem1}
For  $|\mathrm{Arg}\, M^2| \!<\! \frac{\pi}{2}$ and $0\!\le\!\lambda\!\le\!\frac18 (\mathrm{Re}\,M^2)^{3/2}$,
\begin{align}
\sum_{n\ge 0}\big|W^{(n)}_{1,y_1,y_2}[0]\big|\le
\frac{1}{2(c\,\mathrm{Re}\,M^2)^{1/2}}e^{-(c\,\mathrm{Re}\,M^2)^{1/2}|y_1-y_2|}
\end{align}
with $c\!=\!\frac12 \big(1\!+\!\sqrt{1-8\lambda(\mathrm{Re}\,M^2)^{-3/2}}\,\big)$. 
\end{theorem}

This means that the two point function of the full theory has the absolutely convergent expansion  
$\sum_{n\ge 0}W^{(n)}_{1,y_1,y_2}[0]$ and exhibits exponential decay in $|y_1\!-\!y_2|$.
Also, the bound we give here is more explicit than those in previous works (such as \cite{MR08}).

We are also interested in the pressure of the full $\phi_1^4$ theory, which is defined as $\lim_{L\to\infty}\frac{1}{2L}W_{1,L}[0,0]$.
Since $W_{0}[0,0] \!=\! 0$, we have
\begin{align}
\int\!d\nu_{1,L}(\sigma)\,W_{0}[0,\sigma] = -\int_{0}^1\!dt\, \frac{\partial}{\partial t} \!\int\!d\nu_{1-t,L}(\sigma)\,W_{0}[0,\sigma]\notag\\
=\frac12\!\int_{-L}^L\!dx\!\int_{0}^1\!dt\!\int\!d\nu_{1-t,L}(\sigma)\,\frac{\delta^2}{\delta\sigma_x^2}W_{0}[0,\sigma]
\end{align}
and, by \eqref{IntegralEquation},
\begin{align}
W_{1,L}[0,0]=
\frac12\!\int_{-L}^L\!dx\!\int_{0}^1\!dt\!\int\!d\nu_{1-t,L}(\sigma)\,\Big\{\frac{\delta^2}{\delta\sigma_x^2}W_{0}[0,\sigma]+\Big(\frac{\delta}{\delta\sigma_x}W_{t,L}[0,\sigma]\Big)^2\Big\}.
\end{align}
Then, by translation invariance, the pressure of the full theory can be expressed as
\begin{align}\label{Pressure}
&\frac12\lim_{L\to\infty}\!\int_0^1\!\! dt \int\!d\nu_{1-t,L}(\sigma)\,\Big\{\frac{\delta^2}{\delta\sigma_x^2}W_0[0,\sigma]+\Big(\frac{\delta}{\delta\sigma_x}W_{t,L}[0,\sigma]\Big)^2 \Big\}\notag\\
&\!=\frac12\int_0^1\!\! dt\int\!d\nu_{1-t}(\sigma)\,  \Big\{\!-4\lambda (G[\sigma]_{x,x})^2+\Big(\sum_{n\ge 0}W^{(n)}_{t;x}[\sigma]\Big)^2 \Big\}
\end{align}
at least formally, which is independent of $x$.

\begin{theorem}\label{theorem2}
For  $|\mathrm{Arg}\, M^2| \!<\! \frac{\pi}{2}$ and $0\!\le\!\lambda \!\le\!\frac18 (\mathrm{Re}\,M^2)^{3/2}$,
the expression \eqref{Pressure} for the pressure is absolutely convergent and is bounded by
\begin{align}
\frac{\lambda}{\mathrm{Re}\,M^2}\Big\{\frac12+\int_0^1\!\! dt\,\Big({1+\sqrt{1\!-\!8\lambda t(\mathrm{Re}\,M^2)^{-3/2}}}\,\Big)^{-2}\Big\}.
\end{align}
\end{theorem}

\section{Proofs of Theorem 1 and 2}

First we must ensure the well-definedness of $G[\sigma]$ in some sense, which is one of the standpoints of our derivations.
Moverover we need a good bound for it.

\begin{lemma}\label{lemma1}
For $|\mathrm{Arg}\, M^2| \!<\! \frac{\pi}{2}$ and $\lambda\!\ge\! 0$, $G[\sigma]_{x,y}$ is well defined $d\nu_t$-a.e with 
\begin{align}
|G[\sigma]_{x,y}|\le C'_{x,y}:= \frac{1}{2(\mathrm{Re}\,M^2)^{1/2}} e^{-(\mathrm{Re}\,M^2)^{1/2}|x-y|}.
\end{align}
\end{lemma}
\begin{proof}
Using the Wiener integral representation \cite{JL02}, we have, for $\sigma\!\in\! C^\infty(\mathbb{R})\!\cap\!\mathcal{S}'(\mathbb{R})$,
\begin{align}
&\quad\,\,\big|G[\sigma]_{x,y}\big|=\big|((-\Delta\!+\!M^2\!+\!2i\sqrt{2\lambda}\sigma)^{-1})_{x,y}\big|\notag\\
&=\bigg|\int_0^{\infty}\!\!ds\,e^{-sM^2}\!\int\!dW^s_{x,y}(\omega)\,\exp\Big\{\!-\!2i\sqrt{2\lambda}\!\int_0^s\!\! d\xi\,\sigma(\omega(\xi))\Big\}\bigg|\notag\\
&\le\int_0^{\infty}\!\!ds\,e^{-s\,\mathrm{Re}\,M^2}\!\int\!dW^s_{x,y}(\omega)=C'_{x,y},
\end{align}
where $dW^s_{x,y}$ is the conditional Wiener measure on the set of all paths $\omega\!: [0,s]\!\to\!\mathbb{R}$ satisfying $\omega(0)\!=\!x$ and $\omega(s)\!=\!y$.

Let $\eta_\varepsilon(x) \!=\!(2\pi\varepsilon)^{-1/2}e^{-x^2/2\varepsilon}$ for $\varepsilon\!>\!0$,
which satisfies $\partial_\varepsilon \eta_\varepsilon \!=\! \tfrac12 \Delta_x\eta_\varepsilon$.
Then, for general $\sigma\!\in\! \mathcal{S}'(\mathbb{R})$, we have $\sigma * \eta_\varepsilon \!\in\! C^\infty(\mathbb{R})\!\cap\!\mathcal{S}'(\mathbb{R})$ and
\begin{align}
\partial_\varepsilon G[\sigma * \eta_\varepsilon]_{x,y} 
= -i\sqrt{2\lambda}\int\!dz\, G[\sigma * \eta_\varepsilon]_{x,z} G[\sigma * \eta_\varepsilon]_{z,y}(\sigma * \Delta\eta_\varepsilon)_z \notag\\
=-i\sqrt{2\lambda}\int\!dz\, \Delta_z(G[\sigma * \eta_\varepsilon]_{x,z} G[\sigma * \eta_\varepsilon]_{z,y})(\sigma * \eta_\varepsilon)_z.
\end{align}
Assuming $x\!\le\! y$ without loss of generality and using the facts
\begin{align}
\Delta_{x_2} G[\sigma * \eta_\varepsilon]_{x_1,x_2}&= \big(M^2+2i\sqrt{2\lambda}\,(\sigma * \eta_\varepsilon)_{x_2}\big)G[\sigma * \eta_\varepsilon]_{x_1,x_2} - \delta_{x_1,x_2},\\
\partial_{x_2} G[\sigma * \eta_\varepsilon]_{x_1,x_2} &= 
\begin{cases}
\int_{-\infty}^{x_2}\!dz\, \Delta_{z}G[\sigma * \eta_\varepsilon]_{x_1,z} &\text{ for } x_2\!<\!x_1\\
-\int_{x_2}^{\infty}\!dz\, \Delta_{z}G[\sigma * \eta_\varepsilon]_{x_1,z} &\text{ for } x_2\!>\!x_1
\end{cases},
\end{align}
we have $\big|\partial_\varepsilon G[\sigma * \eta_\varepsilon]_{x,y}\big|\le 2\sqrt{2\lambda}\,(Q_1 \! + \! Q_2\! + \! Q_3)$ with
\begin{align}
Q_1 &= \int\!dz\, \big|G[\sigma * \eta_\varepsilon]_{x,z}\big| \big|G[\sigma * \eta_\varepsilon]_{z,y}\big|\big(|M^2|+2\sqrt{2\lambda}\,|(\sigma * \eta_\varepsilon)_z|\big)|(\sigma * \eta_\varepsilon)_z|\notag\\
&\le \int\!dz\, C'_{x,z} C'_{z,y}\big(|M^2|+2\sqrt{2\lambda}\,|(\sigma * \eta_\varepsilon)_z|\big)|(\sigma * \eta_\varepsilon)_z|,\\
Q_2 &= \tfrac12\big|G[\sigma * \eta_\varepsilon]_{x,y}\big|\big(|(\sigma * \eta_\varepsilon)_x|+ |(\sigma * \eta_\varepsilon)_y|\big)\notag\\
&\le \tfrac12 C'_{x,y}\big(|(\sigma * \eta_\varepsilon)_x|+ |(\sigma * \eta_\varepsilon)_y|\big),\\
Q_3 &\le \Big(\int_{-\infty}^{x}\!\!dz\!\int_{-\infty}^z\!\!dz_1\!\int_{-\infty}^z\!\!dz_2
+\int_{x}^{y}\!\!dz\!\int_z^{\infty}\!\!dz_1\!\int_{-\infty}^z\!\!dz_2
+\int_{y}^{\infty}\!\!dz\!\int_z^{\infty}\!\!dz_1\!\int_z^{\infty}\!\!dz_2
\Big)\notag\\
&\;C'_{x,z_1}C'_{z_2,y}\big(|M^2|+2\sqrt{2\lambda}\,|(\sigma * \eta_\varepsilon)_{z_1}|\big)\big(|M^2|+2\sqrt{2\lambda}\,|(\sigma * \eta_\varepsilon)_{z_2}|\big)|(\sigma * \eta_\varepsilon)_z|.
\end{align}
Since
\begin{align}
\int\!d\nu_t(\sigma)\,\prod_{k=1}^n|(\sigma * \eta_\varepsilon)_{x_k}|\le\Big(\int\!d\nu_t(\sigma)\,\prod_{k=1}^n(\sigma * \eta_\varepsilon)_{x_k}^2\Big)^{1/2}\le c(t,n) \varepsilon^{-n/4},
\end{align}
we obtain that, for $\varepsilon_2\!>\!\varepsilon_1\!>\!0$,
\begin{align}
&\int\!d\nu_t(\sigma)\, \big|G[\sigma * \eta_{\varepsilon_2}]_{x,y}-G[\sigma * \eta_{\varepsilon_1}]_{x,y}\big| \notag\\
\le &\int_{\varepsilon_1}^{\varepsilon_2}\!\!d\varepsilon \int\!d\nu_t(\sigma)\, \big|\partial_\varepsilon G[\sigma * \eta_\varepsilon]_{x,y}\big| \le \int_{\varepsilon_1}^{\varepsilon_2}\!\!d\varepsilon\, c(t)\varepsilon^{-3/4},
\end{align}
which goes to 0 as $\varepsilon_2$ goes to 0.
Thus $G[\sigma]_{x,y}:=\lim_{\varepsilon\to 0^+}G[\sigma * \eta_{\varepsilon}]_{x,y}$ exists in $L^1(\mathcal{S}'(\mathbb{R}),d\nu_t)$ and the lemma follows.
\end{proof}

By Lemma \ref{lemma1} and Fubini's theorem, we have that   $G[\sigma_1\!+\!\cdots\!+\!\sigma_n]_{x,y}$ is well defined $d\nu_{t_1}\!\times\!\cdots\!\times\! d\nu_{t_n}$-a.e with the same bound $C'_{x,y}$.



Also we need the following two combinatorial results:

\begin{lemma}\label{lemma2}
\begin{align}
\sum_{K\subset \widehat{N}}\frac{(|K|\!+\!n_1)!}{n_1!}\frac{(N\!-\!|K|\!+\!n_2)!}{n_2!}=\frac{(N\!+\!n_1\!+\!n_2\!+\!1)!}{(n_1\!+\!n_2\!+\!1)!},
\end{align}
where $|K|$ denotes the cardinality of the set $K$.
\end{lemma}
\begin{proof}
Comparing the coefficient of $x^N$ on either side of
\begin{align}
(1\!-\!x)^{-n_1-1}(1\!-\!x)^{-n_2-1}=(1\!-\!x)^{-n_1-n_2-2},
\end{align}
we obtain the combinatorial identity
\begin{align}
\sum_{N_1=0}^N\binom{N_1\!+\!n_1}{N_1}\binom{N\!-\!N_1\!+\!n_2}{N\!-\!N_1}=\binom{N\!+\!n_1\!+\!n_2\!+\!1}{N},
\end{align}
which is equivalent to the required one.
\end{proof}

Let $A_n$ be the $n$th Catalan number \cite{T09}, i.e. $A_0\!=\!1$ and $A_n\!=\!\sum_{h=0}^{n-1}A_h A_{n-h-1}$ for $n\!\ge\! 1$.
Let
\begin{align}
B_{n,m}=
\begin{cases}
1 &\text{if }n\!=\!m\!=\!0\\
\sum_{\substack{n_1,\dots,n_m\ge 1\\n_1+\cdots+n_m=n}}A_{n_1-1}\cdots A_{n_m-1} &\text{if }n\!\ge\!m\!\ge\!1\\
0 &\text{otherwise}
\end{cases}
.
\end{align}

\begin{lemma}\label{lemma3}
The generating function for $(A_n)_{n\ge 0}$ is $\frac{2}{1+\sqrt{1-4x}}$ and the generating function for $(B_{n,m})_{n\ge 0}$ is $\big(\frac{1-\sqrt{1-4x}}{2}\big)^m$. Moreover, for $n,m\!\ge\!1$,
\begin{gather}
B_{n,m}=\frac{1}{n}\sum_{h=0}^{{n}-1}A_{h}\big(m B_{n-h-1,m-1}\!+\!(2n\!-\!2h\!-\!m\!-\!2)B_{n-h-1,m}\big).
\end{gather}
\end{lemma}
\begin{proof}
Let $g(x)\!=\!\frac{2}{1+\sqrt{1-4x}}\!=\!\sum_{n=0}^\infty g_n x^n$. 
Since $g\!=\!1\!+\!xg^2$, we have $g_0\!=\!1$ and $g_n\!=\!\sum_{h=0}^{n-1}g_h g_{n-h-1}$ for $n\!\ge\! 1$.
Thus $g$ is the generating function for $(A_n)_{n\ge 0}$ and $(xg)^m\!=\!\big(\frac{1-\sqrt{1-4x}}{2}\big)^m$ is the generating function for $(B_{n,m})_{n\ge 0}$.
We complete the proof by comparing the coefficient of $x^{n-1}$ on either side of
\begin{align} 
((xg)^m)'&=m(g\!+\!xg')(xg)^{m-1}
=m(g\!+\!x(1\!+\!xg^2)')(xg)^{m-1}\notag\\
&=mg(xg)^{m-1}-mg(xg)^m+2xg((xg)^m)'.
\end{align}

\end{proof}

We are now ready to start the inductions.
First we get a bound for the integrand of $\big|W^{(n)}_{t;z_0,z_1,\dots,z_N}[\sigma]\big|$ in Lemma \ref{lemma4}.
Next $W^{(n)}_{t;z_0,z_1,\dots,z_N}[\sigma]$ is written in terms of $W^{(n,m)}_{t,y_1,y_2;z_1,\dots,z_{N};z'_1,\dots,z'_{N'}}[\sigma]$ in Lemma \ref{lemma5} and then a bound for the integrand of $\big|W^{(n,m)}_{t,y_1,y_2;z_1,\dots,z_{N};z'_1,\dots,z'_{N'}}[\sigma]\big|$ is gotten in Lemma \ref{lemma6}.
All these lemmas are proved inductively.

\begin{lemma}\label{lemma4}
\begin{align}
\int_{\mathbb{R}^N}\!\!dz_1\cdots dz_N \big|W^{(n)}_{t;z_0,z_1,\dots,z_N}[\sigma]\big|\le \frac{t^n}{4^{n+1}}\frac{(2\sqrt{2\lambda})^{N+2n+1}}{(\mathrm{Re}\,M^2)^{N+\frac32 n+\frac12}} \frac{(N\!+\!2n)!}{(2n)!}A_n.
\end{align}
In particular, $\big|W^{(n)}_{t;z_0}[\sigma]\big|\le \dfrac{t^n}{4^{n+1}}\dfrac{(2\sqrt{2\lambda})^{2n+1}}{(\mathrm{Re}\,M^2)^{\frac32 n+\frac12}}A_n$.
\end{lemma}
\begin{proof}
We use induction on $n$. If $n\!=\!0$, 
\begin{align}
&\quad\,\int_{\mathbb{R}^N}\!\!dz_1\cdots dz_N\big|W^{(0)}_{t;z_0,z_1,\dots,z_N}[\sigma]\big|
\le\frac12(2\sqrt{2\lambda})^{N+1}\sum_{\tau\in S_{N}}\notag\\
&\quad\,\int_{\mathbb{R}^N}\!\!dz_1\cdots dz_N\big(\nu_{t}\!*\!\big|
G_{z_0,z_{\tau(1)}}G_{z_{\tau(1)},z_{\tau(2)}}\!\cdots G_{z_{\tau(N)},z_0}\big|\big)[\sigma]\notag\\
&\le\frac12(2\sqrt{2\lambda})^{N+1}\sum_{\tau\in S_{N}}\int_{\mathbb{R}^N}\!\!dz_1\cdots dz_N
C'_{z_0,z_{\tau(1)}}C'_{z_{\tau(1)},z_{\tau(2)}}\!\cdots C'_{z_{\tau(N)},z_0}\notag\\
&\le\frac12(2\sqrt{2\lambda})^{N+1}N!\sup_{z'\in\mathbb{R}}C'_{z',z_0}\int_{\mathbb{R}}\!dz_1\,C'_{z_0,z_1}\cdots\int_{\mathbb{R}}\!dz_N\,C'_{z_{N-1},z_N}\notag\\
&=\frac14\frac{(2\sqrt{2\lambda})^{N+1}}{(\mathrm{Re}\,M^2)^{N+\frac12}}N!.
\end{align}
Assuming it holds for $0\!\le\!n\!\le\!\bar{n}\!-\!1$, we consider the case $n\!=\!\bar{n}$:
\begin{align}
&\quad\,\int_{\mathbb{R}^N}\!\!dz_1\cdots dz_N\big|W^{(\bar{n})}_{t;z_0,z_1,\dots,z_N}[\sigma]\big|\notag\\
&\le\sum_{h=0}^{\bar{n}-1}\sum_{K\subset \widehat{N}}\int_{0}^t\!ds\!\int\!d\nu_{t-s}(\bar{\sigma})\!\int_{\mathbb{R}^{N+1}}\!\!dx dz_1\cdots dz_N\notag\\
&\quad\;\Big|W^{(h)}_{s;x,(z_k)_{k\in K}}[\sigma\!+\!\bar{\sigma}]\Big|\Big|W^{(\bar{n}-h-1)}_{s;z_0,x,(z_k)_{k\in \widehat{N}\backslash K}}[\sigma\!+\!\bar{\sigma}]\Big|\notag\\
&\le\frac{1}{4^{\bar{n}+1}}\frac{(2\sqrt{2\lambda})^{N+2\bar{n}+1}}{(\mathrm{Re}\,M^2)^{N+\frac32 \bar{n}+\frac12}}\sum_{h=0}^{\bar{n}-1}A_{h}A_{\bar{n}-h-1}\!\int_{0}^t\!ds\, s^{\bar{n}-1}\notag\\
&\quad\;(2\bar{n}\!-\!2h\!-\!1)\sum_{K\subset \widehat{N}}\frac{(|K|\!+\!2h)!}{(2h)!}\frac{(N\!-\!|K|\!+\!2\bar{n}\!-\!2h\!-\!1)!}{(2\bar{n}\!-\!2h\!-\!1)!}\notag\\
&=\frac{t^{\bar{n}}}{4^{\bar{n}+1}}\frac{(2\sqrt{2\lambda})^{N+2\bar{n}+1}}{(\mathrm{Re}\,M^2)^{N+\frac32 \bar{n}+\frac12}}\frac{(N\!+\!2\bar{n})!}{(2\bar{n})!}\sum_{h=0}^{\bar{n}-1}\frac{2\bar{n}\!-\!2h\!-\!1}{\bar{n}}A_{h}A_{\bar{n}-h-1}
\end{align}
by the inductive assumption and Lemma \ref{lemma2}.
Since
\begin{align}
\sum_{h=0}^{\bar{n}-1}\frac{2\bar{n}\!-\!2h\!-\!1}{\bar{n}}A_{h}A_{\bar{n}-h-1}
&= \frac12\sum_{h=0}^{\bar{n}-1}\Big(\frac{2\bar{n}\!-\!2h\!-\!1}{\bar{n}}+\frac{2h\!+\!1}{\bar{n}}\Big)A_{h}A_{\bar{n}-h-1}\notag\\
&= \sum_{h=0}^{\bar{n}-1}A_{h}A_{\bar{n}-h-1} = A_{\bar{n}},
\end{align}
we can advance the induction.
\end{proof}

We define recursively 
\begin{align}
&W^{(n,0)}_{t,y_1,y_2;z_1,\dots,z_N;z'_1,\dots,z'_{N'}}[\sigma]=\delta_{n,0}\delta_{N',0}W^{(0)}_{t,y_1,y_2;z_1,\dots,z_N}[\sigma],\\
&W^{(n,m)}_{t,y_1,y_2;z_1,\dots,z_{N};z'_1,\dots,z'_{N'}}[\sigma]=0\\
\intertext{for $m \!>\! n$ and}
&W^{(n,m)}_{t,y_1,y_2;z_1,\dots,z_{N};z'_1,\dots,z'_{N'}}[\sigma]=\sum_{h=0}^{n-1}\sum_{K\subset \widehat{N}'}\int_{0}^t\!ds\!
\int_{\mathbb{R}}\!dx\Big(\nu_{t-s}\!*\!\Big(W^{(h)}_{s;x,(z'_k)_{k\in K}}\notag\\
&\Big(W^{(n-h-1,m-1)}_{s,y_1,y_2;x,z_1,\dots,z_{N};(z'_k)_{k\in \widehat{N}'\backslash K}}+W^{(n-h-1,m)}_{s,y_1,y_2;z_1,\dots,z_{N};x,(z'_k)_{k\in \widehat{N}'\backslash K}}\Big)\Big)\Big)[\sigma]
\end{align}
for $1\!\le\! m \!\le\! n$.
Then we can decompose $W^{(n)}_{t,y_1,y_2;z_1,\dots,z_N}[\sigma]$ as

\begin{lemma}\label{lemma5}
\begin{align}
W^{(n)}_{t,y_1,y_2;z_1,\dots,z_N}[\sigma]=\sum_{m\ge 0}\sum_{K\subset \widehat{N}}W^{(n,m)}_{t,y_1,y_2;(z_k)_{k\in \widehat{N}\backslash K};(z_k)_{k\in K}}[\sigma],
\end{align}
where the number of nonzero summands is finite.
\end{lemma}
\begin{proof}
The lemma is proved by induction on $n$ and is trivial for $n\!=\!0$. Assuming it holds for $n\!=\!0,\dots,\bar{n}\!-\!1$, we consider the case  $n\!=\!\bar{n}$:
\begin{align}
&\quad\,\sum_{m\ge 0}\sum_{K\subset \widehat{N}}W^{(\bar{n},m)}_{t,y_1,y_2;(z_k)_{k\in \widehat{N}\backslash K};(z_k)_{k\in K}}[\sigma]\notag\\
&\!=\sum_{m=1}^{\bar{n}}\sum_{K\subset \widehat{N}}\sum_{h=0}^{\bar{n}-1}\sum_{H\subset K}\int_{0}^t\!ds\!
\int_{\mathbb{R}}\!dx
\Big(\nu_{t-s}\!*\!\Big(W^{(h)}_{s;x,(z_k)_{k\in H}}\notag\\
&\quad\,\Big(W^{(\bar{n}-h-1,m-1)}_{s,y_1,y_2;x,(z_k)_{k\in \widehat{N}\backslash K};(z_k)_{k\in K\backslash H}}
\!+\!W^{(\bar{n}-h-1,m)}_{s,y_1,y_2;(z_k)_{k\in \widehat{N}\backslash K};x,(z_k)_{k\in K\backslash H}}\Big)\Big)\Big)[\sigma]\notag\\
&\!=\sum_{h=0}^{\bar{n}-1}\sum_{H\subset \widehat{N}}\int_{0}^t\!ds\!
\int_{\mathbb{R}}\!dx
\Big(\nu_{t-s}\!*\!\Big(W^{(h)}_{s;x,(z_k)_{k\in H}}\sum_{m\ge 0}\sum_{K:H\subset K\subset \widehat{N}}\notag\\
&\quad\,\Big(W^{(\bar{n}-h-1,m)}_{s,y_1,y_2;x,(z_k)_{k\in \widehat{N}\backslash K};(z_k)_{k\in K\backslash H}}
\!+\!W^{(\bar{n}-h-1,m)}_{s,y_1,y_2;(z_k)_{k\in \widehat{N}\backslash K};x,(z_k)_{k\in K\backslash H}}\Big)\Big)\Big)[\sigma]\notag\\
&\!=\sum_{h=0}^{\bar{n}-1}\sum_{H\subset \widehat{N}}\int_{0}^t\!ds\!
\int_{\mathbb{R}}\!dx
\Big(\nu_{t-s}\!*\!\Big(W^{(h)}_{s;x,(z_k)_{k\in H}}W^{(\bar{n}-h-1,m)}_{s,y_1,y_2;x,(z_k)_{k\in \widehat{N}\backslash H}}\Big)\Big)[\sigma]\notag\\
&\!=W^{(\bar{n})}_{t,y_1,y_2;z_1,\dots,z_N}[\sigma],
\end{align}
which advances the induction.
\end{proof}

\begin{lemma}\label{lemma6}
\begin{align}
&\quad\,\int_{\mathbb{R}^{N+N'}}\!\!dz_1\cdots dz_N dz'_1\cdots dz'_{N'}\big|W^{(n,m)}_{t,y_1,y_2;z_1,\dots,z_{N};z'_1,\dots,z'_{N'}}[\sigma]\big|\le \notag\\
&\frac{t^n}{4^{n}}\frac{(2\sqrt{2\lambda})^{N+N'+2n}}{(\mathrm{Re}\,M^2)^{N'+\frac32 n -m}}
\frac{(N\!+\!m)!}{m!}\frac{(N'\!+\!2n\!-\!m\!-\!1)!}{(2n\!-\!m\!-\!1)!}B_{n,m}(C'^{N+m+1})_{y_1,y_2},
\end{align}
where ${(P\!+\!Q)!}/{Q!}$ can be interpreted as $\prod_{P'=1}^P (P'\!+\!Q)$ when $Q\!<\!0$.
In particular, $\big|W^{(n,m)}_{t,y_1,y_2}[\sigma]\big|\le {(2\lambda t)^n}(\mathrm{Re}\,M^2)^{-\frac32 n +m}
B_{n,m}(C'^{m+1})_{y_1,y_2}$. 
\end{lemma}
\begin{proof}
We use induction on $n$. If $n\!=\!0$, we need only to prove the lemma for $m\!=\!0$ and $N'\!=\!0$:
\begin{align}
&\quad\,\int_{\mathbb{R}^{N}}\!\!dz_1\cdots dz_N\big|W^{(0,0)}_{t,y_1,y_2;z_1,\dots,z_N}[\sigma]\big|
\le(2\sqrt{2\lambda})^{N}\sum_{\tau\in S_{N}}\notag\\
&\quad\,\int_{\mathbb{R}^{N}}\!\!dz_1\cdots dz_N\big(\nu_{t}\!*\!\big|
G_{y_1,z_{\tau(1)}}G_{z_{\tau(1)},z_{\tau(2)}}\!\cdots G_{z_{\tau(N)},y_2}\big|\big)[\sigma]\notag\\
&\le(2\sqrt{2\lambda})^{N}\!\sum_{\tau\in S_{N}}\!\int_{\mathbb{R}^{N}}\!\!dz_1\cdots dz_N\,C'_{y_1,z_{\tau(1)}}C'_{z_{\tau(1)},z_{\tau(2)}}\!\cdots C'_{z_{\tau(N)},y_2}\notag\\
&=(2\sqrt{2\lambda})^{N}N!\,(C'^{N+1})_{y_1,y_2}.
\end{align}
Assuming it holds for $0\!\le\!n\!\le\!\bar{n}\!-\!1$, we consider the case $n\!=\!\bar{n}$ and
need only to prove the lemma for $1\!\le\! m \!\le\! \bar{n}$.
By Lemma \ref{lemma4} and the inductive assumption,
\begin{align}
&\quad\,\int_{\mathbb{R}^{N+N'}}\!\!dz_1\cdots dz_N dz'_1\cdots dz'_{N'}\big|W^{(\bar{n},m)}_{t,y_1,y_2;z_1,\dots,z_{N};z'_1,\dots,z'_{N'}}[\sigma]\big|\notag\\
&\le\sum_{h=0}^{\bar{n}-1}\sum_{K\subset \widehat{N}'}\int_{0}^t\!ds\!\int\!d\nu_{t-s}(\bar{\sigma})\!
\int_{\mathbb{R}^{N+N'+1}}\!\!dx dz_1\cdots dz_N dz'_1\cdots dz'_{N'}\notag\\
&\quad\;\Big|W^{(h)}_{s;x,(z'_k)_{k\in K}}[\sigma\!+\!\bar{\sigma}]\Big|\Big(\Big|W^{(\bar{n}-h-1,m-1)}_{s,y_1,y_2;x,z_1,\dots,z_{N};(z'_k)_{k\in \widehat{N}'\backslash K}}[\sigma\!+\!\bar{\sigma}]\Big|\notag\\
&\quad\;+\Big|W^{(\bar{n}-h-1,m)}_{s,y_1,y_2;z_1,\dots,z_{N};x,(z'_k)_{k\in \widehat{N}'\backslash K}}[\sigma\!+\!\bar{\sigma}]\Big|\Big)\notag\\
&\le\frac{1}{4^{\bar{n}}}\frac{(2\sqrt{2\lambda})^{N+N'+2\bar{n}}}{(\mathrm{Re}\,M^2)^{N'+\frac32 \bar{n} -m}}\frac{(N\!+\!m)!}{m!}(C'^{N+m+1})_{y_1,y_2}\!\int_{0}^t\!ds\, s^{\bar{n}-1}\notag\\
&\quad\;\sum_{h=0}^{\bar{n}-1}A_h\big(m B_{\bar{n}-h-1,m-1}
+(2\bar{n}\!-\!2h\!-\!m\!-\!2)B_{\bar{n}-h-1,m}\big)\notag\\
&\quad\;\sum_{K\subset \widehat{N}'}\frac{(|K|\!+\!2h)!}{(2h)!}\frac{(N'\!-\!|K|\!+\!2\bar{n}\!-\!2h\!-\!m\!-\!2)!}{(2\bar{n}\!-\!2h\!-\!m\!-\!2)!}
\end{align}
 Then using Lemma \ref{lemma2} and \ref{lemma3}, we continue with
\begin{align}
&=\frac{t^{\bar{n}}}{4^{\bar{n}}}\frac{(2\sqrt{2\lambda})^{N+N'+2\bar{n}}}{(\mathrm{Re}\,M^2)^{N'+\frac32 \bar{n} -m}}\frac{(N\!+\!m)!}{m!}\frac{(N'\!+\!2\bar{n}\!-\!m\!-\!1)!}{(2\bar{n}\!-\!m\!-\!1)!}(C'^{N+m+1})_{y_1,y_2}\notag\\
&\quad\;\frac{1}{\bar{n}}\sum_{h=0}^{\bar{n}-1}
A_h\big(m B_{\bar{n}-h-1,m-1}
+(2\bar{n}\!-\!2h\!-\!m\!-\!2)B_{\bar{n}-h-1,m}\big)\notag\\
&=\frac{t^{\bar{n}}}{4^{\bar{n}}}\frac{(2\sqrt{2\lambda})^{N+N'+2\bar{n}}}{(\mathrm{Re}\,M^2)^{N'+\frac32 \bar{n} -m}}\frac{(N\!+\!m)!}{m!}\frac{(N'\!+\!2\bar{n}\!-\!m\!-\!1)!}{(2\bar{n}\!-\!m\!-\!1)!}B_{\bar{n},m}(C'^{N+m+1})_{y_1,y_2},
\end{align}
which advances the induction.
\end{proof}

Combining Lemma \ref{lemma3}, \ref{lemma5} and \ref{lemma6}, we obtain that
\begin{align}
&\quad\;\sum_{n\ge 0}\big|W^{(n)}_{t,y_1,y_2}[\sigma]\big|\le\sum_{n\ge 0}\sum_{m\ge 0}\big|W^{(n,m)}_{t,y_1,y_2}[\sigma]\big|\notag\\
&\le \sum_{m\ge 0}\sum_{n\ge 0}{(2\lambda t)^n}(\mathrm{Re}\,M^2)^{-\frac32 n +m}
B_{n,m}(C'^{m+1})_{y_1,y_2}\notag\\
&=\sum_{m\ge 0}(\varepsilon_t\,\mathrm{Re}\,M^2)^m\!\int_{\mathbb{R}}\!\frac{dp}{2\pi}\frac{e^{i p(y_1-y_2)}}{(p^2\!+\!\mathrm{Re}\,M^2)^{m+1}}=\int_{\mathbb{R}}\!\frac{dp}{2\pi}\frac{e^{i p(y_1-y_2)}}{p^2\!+\!(1\!-\!\varepsilon_t)\mathrm{Re}\,M^2}\notag\\
&
=\frac{1}{2((1\!-\!\varepsilon_t)\mathrm{Re}\,M^2)^{1/2}}e^{-((1-\varepsilon_t)\mathrm{Re}\,M^2)^{1/2}|y_1-y_2|}
\end{align}
for $0\!\le\!\lambda t\!\le\!\frac18(\mathrm{Re}\,M^2)^{\frac32}$, where $\varepsilon_t\!=\!\frac12 \big(1\!-\!\sqrt{1\!-\!8\lambda t(\mathrm{Re}\,M^2)^{-3/2}}\,\big)$.
This completes the proof of Theorem \ref{theorem1}.

By Lemma \ref{lemma1}, $\big|G[\sigma]_{x,x}\big|\le\frac{1}{2(\mathrm{Re}\,M^2)^{1/2}}$. 
By Lemma \ref{lemma3} and \ref{lemma4}, 
\begin{align}
\sum_{n\ge 0}\big|W^{(n)}_{t;x}[\sigma]\big| \le \sum_{n\ge 0} \frac{t^n}{4^{n+1}}\frac{(2\sqrt{2\lambda})^{2n+1}}{(\mathrm{Re}\,M^2)^{\frac32 n+\frac12}} A_n \notag\\
\le \Big(\frac{\lambda}{2\,\mathrm{Re}\,M^2}\Big)^{\frac12}\frac{2}{1+\sqrt{1\!-\!8\lambda t(\mathrm{Re}\,M^2)^{-3/2}}}
\end{align}
for $0\!\le\!\lambda t\!\le\!\frac18(\mathrm{Re}\,M^2)^{\frac32}$.
Then the proof of Theorem \ref{theorem2} follows easily.

\section*{Acknowledgements}
The work is partially supported by Wu Wen-Tsun Key Laboratory of Mathematics.
The author thanks Professor Zheng Yin for valuable advices and suggestions.


\end{document}